\documentclass[review]{elsarticle}

\usepackage{lineno,hyperref}
\modulolinenumbers[5]

\journal{Journal of \LaTeX\ Templates}

\usepackage{natbib}
\usepackage{comment}
\usepackage{hyperref}

\usepackage{amsthm}
\usepackage{amsmath}
\usepackage{graphicx}

\usepackage{algorithm}
\usepackage{algorithmicx}
\usepackage{algpseudocode}

\usepackage[right=3cm,left=3cm,top=1.5cm,bottom=3cm]{geometry}

\usepackage{tikz}
\usetikzlibrary{topaths,calc}
\usetikzlibrary{positioning,chains,fit,shapes,calc}

\newtheorem{theorem}{Theorem}[section]

\newtheorem{proposition}[theorem]{Proposition}

\DeclareMathOperator*{\argmax}{arg\,max}

\begin{document}
\definecolor{myblue}{RGB}{80,80,160}
\definecolor{mygreen}{RGB}{80,160,80}

\begin{frontmatter}
\title{Model-based clustering for random hypergraphs}
\author{Tin-Lok James Ng\fnref{myfootnote}}
\author{Thomas Brendan Murphy\fnref{myfootnote}}
\address{School of Mathematics and Statistics, University College Dublin}
\fntext[myfootnote]{This work was supported by the Science Foundation Ireland funded Insight Research Centre (SFI/12/RC/2289).}
%\author{Tin-Lok James Ng and Thomas Brendan Murphy\thanks{This work was supported by the Science Foundation Ireland funded Insight Research Centre (SFI/12/RC/2289).}}
%\date{School of Mathematics and Statistics, University College Dublin, Ireland}
%\maketitle

\begin{abstract}
%% Text of abstract
A probabilistic model for random hypergraphs is introduced to represent unary, binary and higher order interactions among objects in real-world problems. This model is an extension of the Latent Class Analysis model, which captures clustering structures among objects. An EM (expectation maximization) algorithm with MM (minorization maximization) steps is developed to perform parameter estimation while a cross validated likelihood approach is employed to perform model selection. The developed model is applied to three real-world data sets where interesting results are obtained. 
\end{abstract}


\end{frontmatter}

\section{Introduction}
A large number of random graph models have been proposed \citep{nowicki01, hoff02, handcock07, latouche11} to describe complex interactions among objects of interest. Pairwise relationships among objects can be naturally represented as a graph, in which the objects are represented by the vertices, and two vertices are joined by an edge if certain relationship exists between them. While graphs are capable of representing pairwise interaction between objects, they are inadequate to represent higher order and unary interactions that are typically observed in many real-world problems. Examples of higher-order and unary relationships include co-authorship on academic papers, co-appearance in movie scenes, and songs performed in a concert. 

For example, the study of coauthorship networks of scientists have attracted significant research interests in both natural and social sciences \citep{newman01a, newman01b, newman04, moody04, azondekon18}. Such networks are typically constructed by connecting two scientists if they have coauthored one or more papers together. However, as we will illustrate below, such representation inevitably results in loss of information while a hypergraph representation naturally preserves all information. A hypergraph is a generalization of a graph in which hyperedges are arbitrary sets of vertices, and can contain any number of vertices. As a result, hypergraphs are capable of representing relationships of any arbitrary orders.

We consider a simple example of a coauthorship network with 7 authors and 4 papers in order to illustrate the benefits of hypergraph modelling. A hypergraph representation of the network is given in Figure \ref{fig:hypergraph} where the vertices $v_1, v_2, \ldots, v_7$ represent the authors while the hyperedges $e_1, \ldots, e_4$ represent the papers. For example, the paper $e_1$ is written by four authors $v_1, v_2$, $v_3$ and $v_4$, and the paper $e_2$ is written by two authors $v_2$ and $v_3$, while the paper $e_4$ has a single author $v_4$. 

On the other hand, a graph representation of this coauthorship network with edges between any two authors who have coauthored at least one paper results in the edge set $\{ (v_1, v_2), (v_1, v_3), (v_1, v_4), (v_2, v_3),\\ (v_2, v_4), (v_3, v_4), (v_3, v_5), (v_3, v_6), (v_5, v_6)\}$. It is evident that much information is lost with this representation. In particular, this representation removes information about the number of authors that co-authored a paper. For example, one can only deduce from this edge set that $v_3$ has co-authored with $v_1$ and $v_2$ while unable to conclude that the co-authorship was for the same paper. Furthermore, the hyperedge $e_4$ which contains a singleton $v_4$ is left out in the graph representation.

A number of random hypergraph models were studied in probability and combinatorics literature where theoretical properties such as phase transition, chromatic number were investigated \citep{karonski02, goldschmidt05, panafieu15, dyer15, poole15}. A novel parametrization of distributions on hypergraphs based on geometry of points is proposed in \cite{lunagomez17} which is used to infer Markov structure for multivariate distributions.

On the other hand, statistical modeling with random hypergraph is less explored. \cite{stasi14} introduced the hypergraph beta model with three variants, which is a natural extension of the beta model for random graphs \citep{holland81}. In their model, the probability of a hyperedge $e$ appearing in the hypergraph is parameterized by a vector $ \beta \in \mathbf{R^{N}} $, which represents the ``attractiveness'' of each vertex. However, their model does not capture clustering among objects which is a typical real world phenomenon. In addition, the assumption of an upper bound on the size of hyperedges violates many real world data sets.

One may equivalently represent a hypergraph using a bipartite network (also called two-mode network and affiliation network). Two-mode networks consist of two different kinds of vertices and edges can only be observed between the two types of vertices, but not between vertices of the same type. A hypergraph can be represented as a two-mode by considering the hyperedges as a second type of vertices. For example, an equivalent bipartite representation of the hypergraph shown in Figure \ref{fig:hypergraph} is provided in Figure \ref{fig:bipartite} where the hyperedges $\{ e_1, \ldots, e_4\}$  are now replaced by the four green vertices. 

Two-mode networks have been studied in various disciplines including computer science \citep{perugini04}, social sciences \citep{faust02, robins04, koskinen12, friel16} and physics \citep{lind05}. A number of approaches have been proposed to analyze and model two-mode network data \citep{borgatti97, robins04, doreian04, latapy08, wang09, snijders13}. In particular, models originally developed for binary networks were extended for two-mode networks.

\cite{doreian04} developed a blockmodeling approach of two-mode network data which aims to simultaneously partition the two types of vertices into blocks. \cite{skvoretz99} proposed exponential random graph models (ERGMs) for two-mode networks which models the logit of the probability of an actor belong to an event as a function of actor and event specific effects and other graph statistics. A clustering algorithm for two-mode network is developed in \cite{field06} based on the modelling framework in \cite{skvoretz99}. Several extensions to the ERGMs for bipartite networks are proposed in recent years \citep{wang09, wang13}. \cite{snijders13} proposed a methodology for studying the co-evolution of two-mode and one-mode networks. A network autocorrelation model for two-mode networks is introduced in \cite{fujimoto11}.

Representing network observations using two-mode networks has the benefits of modelling vertices of both types jointly. However, in analyzing a two-mode network, one type of vertices may attract most interest. For example, in co-authorship networks, the main interest may lie in the collaborations rather than in co-authored papers. In such scenarios, a hypergraph representation is most nature by converting one type of vertices into hyperedges with no loss of information. 

In this paper, we propose the Extended Latent Class Analysis (ELCA) model for random hypergraphs, which is a natural extension of the Latent Class Analysis (LCA) model \citep{lazarsfeld68, goodman74, celeux91} and includes the LCA model as a special case. The model is applied to two applications, including Star Wars movie scenes and Lady Gaga concerts 2014.

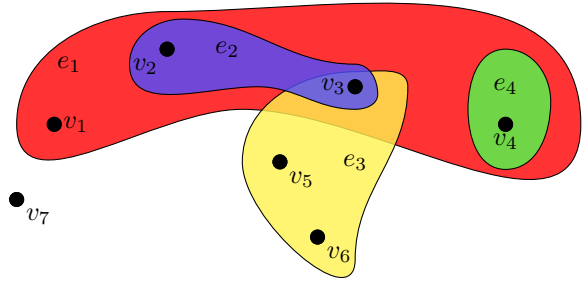
\begin{figure}
\centering
\begin{tikzpicture}
    \node (v1) at (0,2) {};
    \node (v2) at (1.5,3) {};
    \node (v3) at (4,2.5) {};
    \node (v4) at (6,2) {};
    \node (v5) at (3,1.5) {};
    \node (v6) at (3.5,.5) {};
    \node (v7) at (-0.5,1) {};

    \begin{scope}[fill opacity=0.8]
    \filldraw[fill=red!180] ($(v1)+(-0.5,0)$) 
        to[out=90,in=180] ($(v2) + (0,0.5)$) 
        to[out=0,in=90] ($(v4) + (1,0)$)
        to[out=270,in=0] ($(v2) + (1,-0.8)$)
        to[out=180,in=270] ($(v1)+(-0.5,0)$);
    \filldraw[fill=green!70] ($(v4)+(-0.5,0.2)$)
        to[out=90,in=180] ($(v4)+(0,1)$)
        to[out=0,in=90] ($(v4)+(0.6,0.3)$)
        to[out=270,in=0] ($(v4)+(0,-0.6)$)
        to[out=180,in=270] ($(v4)+(-0.5,0.2)$);
    \filldraw[fill=yellow!80] ($(v5)+(-0.5,0)$)
        %to[out=90,in=180] ($(v3)+(-0.5,-1)$)
        %to[out=45,in=270] ($(v3)+(-0.7,0)$)
        to[out=90,in=180] ($(v3)+(0,0.2)$)
        to[out=0,in=90] ($(v3)+(0.7,0)$)
        %to[out=270,in=90] ($(v3)+(-0.3,-1.8)$)
        to[out=270,in=90] ($(v6)+(0.5,-0.3)$)
        to[out=270,in=270] ($(v5)+(-0.5,0)$);
    \filldraw[fill=blue!70] ($(v2)+(-0.5,-0.2)$) 
        to[out=90,in=180] ($(v2) + (0.2,0.4)$) 
        to[out=0,in=180] ($(v3) + (0,0.3)$)
        to[out=0,in=90] ($(v3) + (0.3,-0.1)$)
        to[out=270,in=0] ($(v3) + (0,-0.3)$)
        to[out=180,in=0] ($(v3) + (-1.3,0)$)
        to[out=180,in=270] ($(v2)+(-0.5,-0.2)$);
    \end{scope}

    \foreach \v in {1,2,...,7} {
        \fill (v\v) circle (0.1);
    }

    \fill (v1) circle (0.1) node [right] {$v_1$};
    \fill (v2) circle (0.1) node [below left] {$v_2$};
    \fill (v3) circle (0.1) node [left] {$v_3$};
    \fill (v4) circle (0.1) node [below] {$v_4$};
    \fill (v5) circle (0.1) node [below right] {$v_5$};
    \fill (v6) circle (0.1) node [below right] {$v_6$};
    \fill (v7) circle (0.1) node [below right] {$v_7$};

    \node at (0.2,2.8) {$e_1$};
    \node at (2.3,3) {$e_2$};
    \node at (4,1.5) {$e_3$};
    \node at (6,2.5) {$e_4$};
\end{tikzpicture}
\caption{A hypergraph representation of a coauthorship network.}
\label{fig:hypergraph}
\end{figure}

\begin{figure}
\centering
\begin{tikzpicture}[thick,
  every node/.style={draw,circle},
  fsnode/.style={fill=myblue},
  ssnode/.style={fill=mygreen},
  every fit/.style={ellipse,draw,inner sep=-2pt,text width=2cm},
  ->,shorten >= 3pt,shorten <= 3pt
]

% the vertices of U
\begin{scope}[start chain=going below,node distance=5mm]
\foreach \i in {1,2,...,7}
  \node[fsnode,on chain] (f\i) [label=left: \i] {};
\end{scope}

% the vertices of V
\begin{scope}[xshift=4cm,yshift=-0.5cm,start chain=going below,node distance=7mm]
\foreach \i in {1,...,4}
  \node[ssnode,on chain] (s\i) [label=right: \i] {};
\end{scope}

% the set U
\node [myblue,fit=(f1) (f7),label=above:$V$] {};
% the set V
\node [mygreen,fit=(s1) (s4),label=above:$E$] {};

% the edges
\draw (f1) -- (s1);
\draw (f2) -- (s1);
\draw (f3) -- (s1);
\draw (f4) -- (s1);
\draw (f2) -- (s2);
\draw (f3) -- (s2);
\draw (f3) -- (s3);
\draw (f5) -- (s3);
\draw (f6) -- (s3);
\draw (f4) -- (s4);
\end{tikzpicture}
\caption{Bipartite graph representation of the hypergraph in Figure \ref{fig:hypergraph}.}
\label{fig:bipartite}
\end{figure}
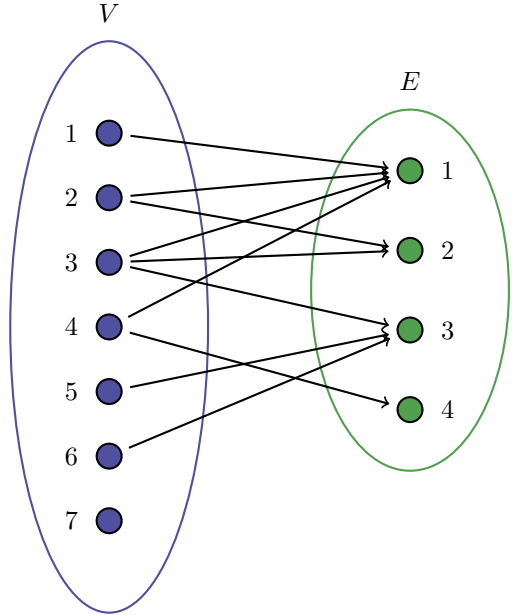

\section{Model and Motivation}
\subsection{Hypergraph}
A hypergraph is represented by a pair $H = (V,E)$, where $V = \{V_{1},V_{2},\cdots,V_{N}\}$ is the set of $N$ vertices and $E = \{e_{1}, e_{2}, \cdots, e_{M}\}$ is the set of $M$ hyperedges. A hyperedge $e$ is a subset of $V$, and we allow repetitions in the hyperedge set $E$. Thus, the hypergraph $H$ can alternatively be represented with a $ N \times M $ matrix $ \mathbf{x} = (x_{ij})$ where $ x_{ij} = 1 $ if vertex $V_{i}$ appears in hyperedge $e_{j}$ and $ x_{ij}=0$ otherwise.

\subsection{Latent Class Analysis Model for Random Hypergraphs}
The binary latent class analysis (LCA) model \citep{lazarsfeld68, goodman74} is a commonly used mixture model for high dimensional binary data. It assumes that each observation is a member of one and only one of the $G$ latent classes, and conditional on the latent class membership, the manifest variables are mutually independent of each other. The LCA model appears to be a natural candidate to model random hypergraphs where hyperedges are partitioned into $G$ latent classes, and the probability that a hyperedge $e \in E$ contains a vertex $v \in V$ depends only on its latent class assignment. 

Let $ \pi = (\pi_{1}, \cdots, \pi_{G} ) $ be the {\em a priori} latent class assignment probabilities where $G$ is the number of latent classes,  and define the $N \times G$ matrix $ p = (p_{ig})$ and $p_{ig}$ is the probability that vertex $V_{i}$ is contained in a hyperedge $e$ with latent class label $g$. The likelihood function can be written as
\begin{eqnarray*}
  L(\mathbf{x}; p, \pi) = \prod_{j=1}^{M} \Big[ \sum_{g=1}^{G} \pi_{g} \prod_{i=1}^{N} p_{ig}^{x_{ij}} (1-p_{ig})^{1-x_{ij}} \Big].
\end{eqnarray*}

By introducing the $M \times G$ latent class membership matrix $\mathbf{z}^{(1)} = (z^{(1)}_{jg})$ where $z^{(1)}_{jg} = 1$ if hyperedge $e_{j}$ has latent class label $g$ and $z^{(1)}_{jg}=0$ otherwise, the complete data likelihood of $\mathbf{x}$ and $\mathbf{z}^{(1)}$ can be expressed as (\ref{lca_lik}).
\begin{eqnarray}
  \label{lca_lik}
  L(\mathbf{x}, \mathbf{z}^{(1)}; p, \pi) = \prod_{j=1}^{M} \prod_{g=1}^{G} 
             \Big[ \pi_{g} \prod_{i=1}^{N} p_{ig}^{x_{ij}} (1-p_{ig})^{1-x_{ij}} \Big]^{z^{(1)}_{jg}}.
\end{eqnarray}

In comparison to the hypergraph beta models introduced in \cite{stasi14}, the LCA model is capable of capturing the clustering and heterogeneity of hyperedges. For example, academic papers can be naturally labelled according to subject areas and conditional on a paper being labelled mathematics, one would expect that the probability a mathematician co-authored the paper is higher than a biologist. The LCA model does not assume an upper bound on the size of hyperedges and can model hyperedges of any size. Furthermore, an efficient expectation maximization algorithm \citep{dempster77} can be easily derived to perform parameter estimation.

\subsection{Extended Latent Class Analysis for Random Hypergraphs}
While the LCA model captures the clustering and heterogeneity of hyperedges in real world data sets, it is quite restrictive in modeling the size of a hyperedge. The size of a hyperedge $e$ with latent class label $g$ follows the Poisson Binomial distribution \citep{wang93} with parameters $ (p_{1g}, \cdots, p_{Ng})$, and with expected value $\sum_{i=1}^{N} p_{ig}$ and variance $\sum_{i=1}^{N} p_{ig} (1-p_{ig})$. As we will illustrate in a few real world data sets, the LCA model underestimates the variation in sizes of hyperedges. Thus, we extend the LCA model by including an additional clustering structure to address this shortcoming. 

We develop the Extended Latent Class Analysis model (ELCA) by introducing an additional clustering to the hyperedges. We assume that the two clustering are independent. We let $ \tau=(\tau_{1},\cdots,\tau_{K})$ be the {\em a priori} additional clustering assignment probabilities where $K$ is the number of additional clusters.  Thus, the probability that a hyperedge has cluster label $g$ and additional cluster label $k$ is given by $ \pi_{g} \tau_{k}$. We define the $N \times G$ matrix $\phi = (\phi_{ig})$ and $K$ dimensional vector $a=(a_{1}, \cdots, a_{K})$ so that the probability that vertex $V_{i}$ is contained in a hyperedge with cluster label $g$ and additional cluster label $k$ is given by $ a_{k} \phi_{ig} $.

Let $\theta = (\pi, \tau, \phi, a)$ denote the model parameters, the likelihood function can be written as
\begin{eqnarray*}
  L(\mathbf{x};\theta) = \prod_{j=1}^{M} \Big[ \sum_{g=1}^{G} \sum_{k=1}^{K} \pi_{g} \tau_{k} \prod_{i=1}^{N} (a_{k} \phi_{ig})^{x_{ij}} (1-a_{k} \phi_{ig})^{1-x_{ij}} \Big].
\end{eqnarray*}

We define the $M \times K$ additional cluster membership matrix $\mathbf{z}^{(2)} = (z_{jk}^{(2)})$ where $z_{jk}^{(2)} = 1$ if hyperedge $e_{j}$ has additional cluster label $k$ and $z_{jk}^{(2)}=0$ otherwise. The complete data likelihood function of $\mathbf{x}$, $\mathbf{z}^{(1)}$ and $\mathbf{z}^{(2)}$ is given as:
\begin{eqnarray}
  \label{glca_lik}
   L(\mathbf{x},\mathbf{z}^{(1)}, \mathbf{z}^{(2)};\theta) = \prod_{j=1}^{M} \prod_{g=1}^{G} \prod_{k=1}^{K} 
  \Big[ \pi_{g} \tau_{k} \prod_{i=1}^{N} (a_{k} \phi_{ig})^{x_{ij}} (1-a_{k}\phi_{ig})^{1-x_{ij}} \Big]^{z^{(1)}_{jg} z^{(2)}_{jk}}.
\end{eqnarray}

We further impose the constraint $a_{K} = 1$ to ensure that the model is identifiable. It is easy to see that the LCA model is a special case of the ELCA model by letting the number of additional clusters $K=1$.

\subsection{Theoretical Properties}
We compare the theoretical properties of the LCA and ELCA models developed above. Proposition~\ref{comp_glca} below shows that the size of hyperedge simulated from the ELCA model has larger variance than simulated from the LCA model.  

\begin{proposition}
\label{comp_glca}
  Suppose we are given the LCA model with parameters $ \{ \pi, p \}$ and the ELCA model with parameters $ \{ \pi, \tau, a, \phi \}$ and $N$ vertices. Suppose the condition $p_{ig} = \phi_{ig} \sum_{k=1}^{K} a_{k} \tau_{k} $ holds for $i=1,\cdots,N$ and $g=1,\cdots, G$. 

 Let $A$ denote the size of a random hyperedge $X_{A}$ generated under the LCA model. Similarly, let $B$ denote the size of a random hyperedge $X_{B}$ generated under the ELCA model. We have the following results.
\begin{eqnarray*}
   E(A) = E(B)
\end{eqnarray*} 
\begin{eqnarray*}
   Var(A) \le Var(B)
\end{eqnarray*}
\end{proposition}

\begin{proof}
  The proof is straightforward and is given in the Appendix.
\end{proof}

We now let $f_{N}(y)$ be the probability mass functions of the size of a random hyperedge simulated from a $G$ cluster LCA model. Similarly, we let $h_{N}(y)$ be the probability mass function of the size of a random hyperedge simulated from the ELCA model with $G$ clusters and $K$ additional clusters. The following result can be derived.

\begin{proposition}
 \label{conv_lca}
\begin{enumerate}
\item
  Under the specifications of a LCA model with parameters $\pi=(\pi_{1},\cdots,\pi_{G})$ and $ \{ p_{ig} \}_{i=1,\cdots,N, g=1,\cdots,G} $, and suppose the following conditions hold for $g=1,\cdots,G$,
  \begin{eqnarray*}
     \lambda_{N}^{(g)} = \sum_{i=1}^{N} p_{ig} \rightarrow \lambda^{(g)} > 0 
  \end{eqnarray*}
  \begin{eqnarray*}
     \sum_{i=1}^{N} p_{ig}^{2} \rightarrow 0 
  \end{eqnarray*}
as $N \rightarrow \infty$. We have
\begin{eqnarray*}
   f_{N}(y) \rightarrow \sum_{g=1}^{G} \pi_{g} \frac{ e^{ -\lambda^{(g)} } (\lambda^{(g)})^{y} }{y!}
\end{eqnarray*}
That is, the distribution of the size of a random hyperedge converges to a mixture of Poisson distribution with $G$ components.

\item  Under the specification of a ELCA model with parameters $\pi=(\pi_{1},\cdots,\pi_{G})$, $ \tau=( \tau_{1}, \cdots, \tau_{K})$, $a=(a_{1}, \cdots, a_{K})$, and $ \{ \phi_{ig} \}_{i=1,\cdots,N, g=1,\cdots,G} $.  Further suppose the following conditions hold for $g=1,\cdots,G$, and $k=1, \cdots, K$.
\begin{eqnarray*}
     \lambda_{N}^{(g,k)} = \sum_{i=1}^{N} \phi_{ig} a_{k} \rightarrow \lambda^{(g,k)} > 0 
  \end{eqnarray*}
  \begin{eqnarray*}
     \sum_{i=1}^{N} \phi_{ig}^{2} a_{k}^{2} \rightarrow 0 
  \end{eqnarray*}
as $N \rightarrow \infty$. We have

\begin{eqnarray*}
   h_{N}(y) \rightarrow \sum_{g=1}^{G} \sum_{k=1}^{K} \pi_{g} \tau_{k} \frac{ e^{ - \lambda^{(g,k)} } ( \lambda^{(g, k)})^{y} }{y!}
\end{eqnarray*}

That is, the distribution of the size of a random hyperedge converges to a mixture of Poisson distribution with $G \times K$ components.
\end{enumerate}
\end{proposition}

\begin{proof}
  Conditional on the event that a random hyperedge is generated from cluster $g$, \cite[][Theorem 3]{wang93} implies that
  \begin{eqnarray*}
     f_{N}(y) \rightarrow \frac{e^{-\lambda^{(g)}} (\lambda^{(g)})^{y} }{ y! }
  \end{eqnarray*}
Part 1 result follows by marginalizing over the $G$ clusters. The second part of the proposition can be proved similarly.
\end{proof}

Proposition~\ref{conv_lca} implies that the size distribution of a random hyperedge generated under the ELCA model is far more flexible than for the LCA model.

\subsection{Co-Clustering}
The concept of having two clustering structure is related to co-clustering or block clustering. In co-clustering, the objective is to simultaneously cluster rows and columns of a data matrix. In particular, mixture models have been proposed with EM algorithms developed in the context of co-clustering \citep{govaert03, govaert08}. Co-clustering has also received significant attention in various application such as text mining, bioinformatics and recommender systems \citep{dhillon03, cheng00, george05}. In comparison, we aim to obtain two types of clustering structure for the rows of a data matrix.

In the work of \cite{rau15}, a Poisson mixture model was proposed for clustering of digital gene expression to discover groups of co-expressed genes, where observations of biological entities under different conditions are collected. In order to model the variations in overall expression level among biological entities, a scaling parameter is introduced for each entity. In comparison, we explicitly model the size of random hyperedge using clustering which results in a more parsimonious model structure.

\section{EM Algorithm}
\label{sec_em}

We estimate the parameters $\theta = (\pi, \tau, \phi, a) $ of the ELCA model using an EM algorithm \citep{dempster77} which is a popular method in fitting mixture models. The E-step of the EM algorithm involves computing the expected value of the logarithm of the complete data likelihood (\ref{glca_lik}) with respect to the distribution of the unobserved $\mathbf{z}^{(1)}$ and $\mathbf{z}^{(2)}$ given the current estimates. The M-step involves maximizing the expected complete data log-likelihood.

Taking logarithm of the complete data likelihood in (\ref{glca_lik}), we obtain the complete data log-likelihood function below.
\begin{eqnarray}
  \label{glca_ll}
   \log L(\mathbf{x},\mathbf{z};\theta) = \sum_{j=1}^{M} \sum_{g=1}^{G} \sum_{k=1}^{K} Z^{(1)}_{jg} Z^{(2)}_{jk} \Big[ \log \pi_{g} + \log \tau_{k} + \sum_{i=1}^{N} x_{ij} \log(a_{k}) +  \nonumber
\\ \log(\phi_{ig}) + (1-x_{ij}) \log(1-a_{k} \phi_{ig}) \Big]. 
\end{eqnarray}

For the E-step, we need to evaluate the expectation of (\ref{glca_ll}) conditional on data $x$ and current parameter estimates $\theta^{(t)}$.
\begin{eqnarray*}
  Q(\theta|\theta^{(t)}) := E(\log L(x,z;\theta)|x,\theta^{(t)})
\end{eqnarray*}
That is, we need to evaluate the expectation $ \widehat{Z^{(1)}_{jg} Z^{(2)}_{jk}} := E(Z^{(1)}_{jg} Z^{(2)}_{jk}|x, \theta^{(t)})$. We have that
\begin{eqnarray}
\label{update_e_step}
  E(Z^{(1)}_{jg} Z^{(2)}_{jk} | x, \theta^{(t)}) &=& Pr(Z^{(1)}_{jg} = Z^{(2)}_{jk} = 1 | x, \theta^{(t)})   \nonumber \\  
               &= & \frac{ \pi^{(t)}_{g} \tau^{(t)}_{k} \Big[ \prod_{i=1}^{N} (a_{k} \phi_{ig})^{x_{ij}} (1-a_{k}\phi_{ig})^{1-x_{ij}} \Big]}{
  \sum_{g=1}^{G} \sum_{k=1}^{K} \pi^{(t)}_{g} \tau^{(t)}_{k} \Big[ \prod_{i=1}^{N} (a_{k} \phi_{ig})^{x_{ij}} (1-a_{k}\phi_{ig})^{1-x_{ij}} \Big] }.
\end{eqnarray}
In particular, the E-step has a computational complexity of ${\cal O}(N)$ for each pair $(g,k)$. While the E-step of the EM algorithm is straightforward, the M-step involves complicated maximization. Thus, we use the ECM algorithm \citep{meng93} which replaces the complex M-step by a series of simpler conditional maximizations. The conditional maximizations with respect to the parameters $\phi$ and $a$ do not have closed form solutions. We resort to the MM algorithm \citep{lange00, hunter04} which works by lower bounding the objective function by a minorizing function and then maximizing the minorizing function. Details of the M-step are given in the appendix and the EM algorithm is summarized in Algorithm \ref{algo_em}. In particular, we note that the computational complexity for maximizing $ \phi_{ig} $ and $ a_{k} $ are given by ${\cal O}(N_{iter}MK)$ and ${\cal O}(N_{iter}MGN)$, respectively, where $N_{iter}$ is the number of iterations required for the MM algorithm.

\begin{algorithm}
  \caption{EM Algorithm \newline \textbf{Input}: $ \mathbf{x}, G, K, tol $  \newline \textbf{Output}:  $ \hat{\phi}, \hat{a}, \hat{\pi}, \hat{\tau}, \mathbf{\hat{z}}^{(1)}, \mathbf{\hat{z}}^{(2)}$} 
  \begin{algorithmic}[1]
        \State $conv = False$
        \State Random initialization of $\phi, a, \pi, \tau$
        \While{$conv = False$}
            \State Do the E-step according to (\ref{update_e_step})
            \For{$i=1,\cdots,N$}
                \For{$g=1,\cdots,G$}
                    \State Update $ \phi_{ig} $ according to (\ref{update_phi})
                \EndFor
            \EndFor
            \For{$k=1,\cdots,K-1$}
                  \State Update $ a_{k} $ according to (\ref{update_a})
            \EndFor
            \For{$g=1,\cdots,G$}
                  \State Update $ \pi_{g} $ according to (\ref{update_pi})
            \EndFor
            \For{$k=1,\cdots,K$}
                  \State Update $ \tau_{k} $ according to (\ref{update_tau})
            \EndFor
            \State Evaluate Change in log-likelihood $ \Delta_{loglik} $ resulting from parameter updates
            \If{ $ \Delta_{loglik} < tol $}
                 \State $conv = True$
            \EndIf
        \EndWhile
   \end{algorithmic}
   \label{algo_em}
\end{algorithm}

\section{Model Selection}

\subsection{Cross Validated Likelihood}
\label{sec_cv}
Given a fixed model, the cross validated likelihood method \citep{smyth00} works by repetitively partitioning the observations into two disjoint sets, one of which is used to fit the model and obtain estimates of model parameters by maximizing the log-likelihood, and the other is for evaluating the model by computing its log-likelihood.

For each $G$ and $K$, we define ${\cal M}_{G,K}$ to be the ELCA model with $G$ clusters and $K$ additional clusters. To apply the cross validated likelihood method, we randomly partition the hyperedges $\mathbf{x}$ into two sets $\mathbf{x}^{(train)}$ and $\mathbf{x}^{(test)}$ where each hyperedge in $\mathbf{x}$ is included in $\mathbf{x}^{(train)}$ with probability $q$. In our applications we set $q = 0.7$. The EM algorithm developed in section \ref{sec_em} is then used to fit $\mathbf{x}^{(train)}$ and obtain the parameter estimates $ \hat{\theta} = (\hat{\pi}, \hat{\tau}, \hat{\phi}, \hat{a}) $. We then compute the log-likelihood of $\mathbf{x}^{(test)}$ under the estimated parameters $\hat{\theta}$ and obtain the test log-likelihood $L^{(test)}$. The above procedure is then repeated $N_{cv}$ times and the estimated cross validated log-likelihood is obtained by averaging over $L^{(test)}$. The procedure above is summarized in Algorithm \ref{algo_cv}.

We perform a greedy search for the optimal combination of $G$ and $K$ which produces the largest estimated cross validated log-likelihood $\hat{L}_{cv}^{G, K}$. Starting with one cluster and one additional cluster, an additional cluster is then successively added to the model until the estimated cross validated log-likelihood does not increase. At this stage, we then increment the number of clusters $G$ by 1 and the above procedure is repeated provided that $\hat{L}_{cv}^{G,1} > \hat{L}_{cv}^{G-1,1}$. The greedy search algorithm is summarized in Algorithm \ref{algo_greedy}. The greedy search can be computationally intensive when the search space for $(G,K)$ and the number of cross validation $N_{cv}$ are large. 

\begin{algorithm}
  \caption{Estimated Cross Validated Log-likelihood \newline \textbf{Input}: $ \mathbf{x}, G, K, N_{cv}, q$  \newline \textbf{Output}:  $ \hat{L}_{cv}$} 
  \begin{algorithmic}[1]
     \For{\texttt{ $ n = 1,\cdots, N_{cv}$}}
         \State \texttt{ $\mathbf{x}^{(train)} = \emptyset$, $\mathbf{x}^{(test)} = \emptyset$ }
         \For{\texttt{ $j=1, \cdots, M$}}
             \State \texttt{ $ u \sim Unif(0,1) $}
             \If{$u < q$}
                 \State \texttt{ $\mathbf{x}^{(train)} = \mathbf{x}^{(train)} \cup x_{j} $}
             \Else
                 \State \texttt{$\mathbf{x}^{(test)} = \mathbf{x}^{(test)} \cup x_{j} $}
             \EndIf
         \EndFor
         \State \texttt{ $ \hat{\theta} = \argmax_{\theta} \{ L(\mathbf{x}^{(train)}; \theta) \} $ }  
         \State \texttt{ $ \hat{L}_{n} = L(\mathbf{x}^{(test)}; \theta) $ }
     \EndFor
     \State \texttt{  $ \hat{L}_{cv} = \frac{1}{N_{cv}} \sum_{n=1}^{N_{cv}} \hat{L}_{n} $ }
   \end{algorithmic}
   \label{algo_cv}
\end{algorithm}

\begin{algorithm}
  \caption{Greedy Search For Model Selection \newline \textbf{Input}: $ \mathbf{x} $ \newline \textbf{Output}: $ G_{opt}, K_{opt} $ }
  \begin{algorithmic}[1]
       \State $ G_{opt} = 1 $
       \State $ K_{opt} = 1 $
       \State $stop_{G} = False$
       \State Obtain $ \hat{l}_{cv}^{G_{opt},K_{opt}}$ using Algorithm \ref{algo_cv}
       \While{$stop_{G} = False$}
           \State $stop_{K} = False$
           \While{$stop_{K} = False$}   
               \State $K_{test} = K_{opt} + 1$
               \State Obtain $ \hat{l}_{cv}^{G_{opt},K_{test}} $ using Algorithm \ref{algo_cv}
               \If{ $\hat{l}_{cv}^{G_{opt},K_{test}} > \hat{l}_{cv}^{G_{opt},K_{opt}} $}
                    \State $K_{opt} = K_{test}$
               \Else
                    \State $stop_{K} = True$
               \EndIf
            \EndWhile
            \State $G_{test} = G_{opt} + 1$
            \State Obtain $\hat{l}_{cv}^{G_{test},1}$ using Algorithm \ref{algo_cv}
            \If{$\hat{l}_{cv}^{G_{test},1} > \hat{l}_{cv}^{G_{opt},1}$}
               \State $ G_{opt} = G_{test} $
            \Else
               \State $Stop_{G} = True$
            \EndIf
       \EndWhile
   \end{algorithmic}
  \label{algo_greedy}
\end{algorithm}

\section{Applications}

\subsection{Star Wars Movie Scenes}

Our first application is modeling co-appearance of the main characters in the scenes of the movie ``Star Wars: A New Hope''. We collected the scripts of the movie from The Internet Movie Script Database \footnote{\label{starwars}Movie script data freely available at {\tt https://www.imsdb.com/} } and constructed a hypergraph for the eight main characters. We define each scene in the movie as a hyperedge with a total of 178 hyperedges, and a character is contained in the scene if he/she speaks in the scene. 

We first performed model selection using the greedy search algorithm and the cross validated likelihood method presented in Section \ref{sec_cv} to select the optimal number of clusters and additional clusters for the ELCA model. The results of the greedy search are provided in Table~\ref{ms_starwars} and the model with 3 clusters and 2 additional clusters is selected. 

The results from fitting the ELCA model with $G=3$ and $K=2$ are provided in Table~\ref{starwar_para_est} and Table~\ref{starwar_phi}. We can see the variation in the size of hyperedges from the parameter estimates $\hat{a}$ and $\hat{\tau}$ with the majority ($81 \%$) of hyperedges having size much smaller than the rest of the hyperedges. Thus, one can deduce that a small proportion of the movie scenes have far more characters.

\begin{table}[hbtp]
 \caption{Estimates of $ \pi $, $ \tau $ and $a$ from fitting the ELCA model with 3 clusters and 2 additional clusters for the Star Wars data set}
\begin{center}
\begin{tabular}{| c | c |}
 \hline
  $ \hat{\pi} $ & $(0.40, 0.40, 0.20)$ \\ \hline
  $ \hat{\tau} $ & $(0.81, 0.19)$ \\ \hline
  $ \hat{a} $ & $(0.41, 1)$ \\ \hline
\end{tabular}
\end{center}
  \label{starwar_para_est}
\end{table}

\begin{table}[hbtp]
\caption{Estimates of $\{ \phi_{ig} \}$ from fitting the ELCA model with 3 clusters and 2 additional clusters for the Star Wars data set}
\begin{center}
\begin{tabular}{|c|c|c|c|}
   \hline
Character & Cluster 1 & Cluster 2 & Cluster 3  \\ \hline
   Wedge &  0.18 & 0.00 & 0.36 \\ \hline
   Han &  0.00 & 1.00 & 0.00 \\ \hline
   Luke & 1.00 & 1.00 & 0.00  \\ \hline
   C-3PO &  0.75 & 0.30 & 0.00\\ \hline
   Obi-Wan &  0.00 & 0.00 & 1.00 \\ \hline
   Leia & 0.12 & 0.48 & 0.07 \\ \hline
   Biggs & 0.31 & 0.00 & 0.28 \\ \hline
   Darth Vader & 0.19 & 0.35 & 0.06 \\ \hline
\end{tabular}
\end{center}
\label{starwar_phi}
\end{table}

The estimates $\hat{\phi}$ in Table~\ref{starwar_phi} reveal interesting clustering structure for the 8 main characters in the movie. For example, the lead character ``Luke'' has a strong tendency to appear in the two largest clusters. On the other hand, it is extremely unlikely for ``Obi-Wan'' and ``Han'' appear in the same scene. 

\begin{figure}
  \centering
  \includegraphics[width=1.0\textwidth]{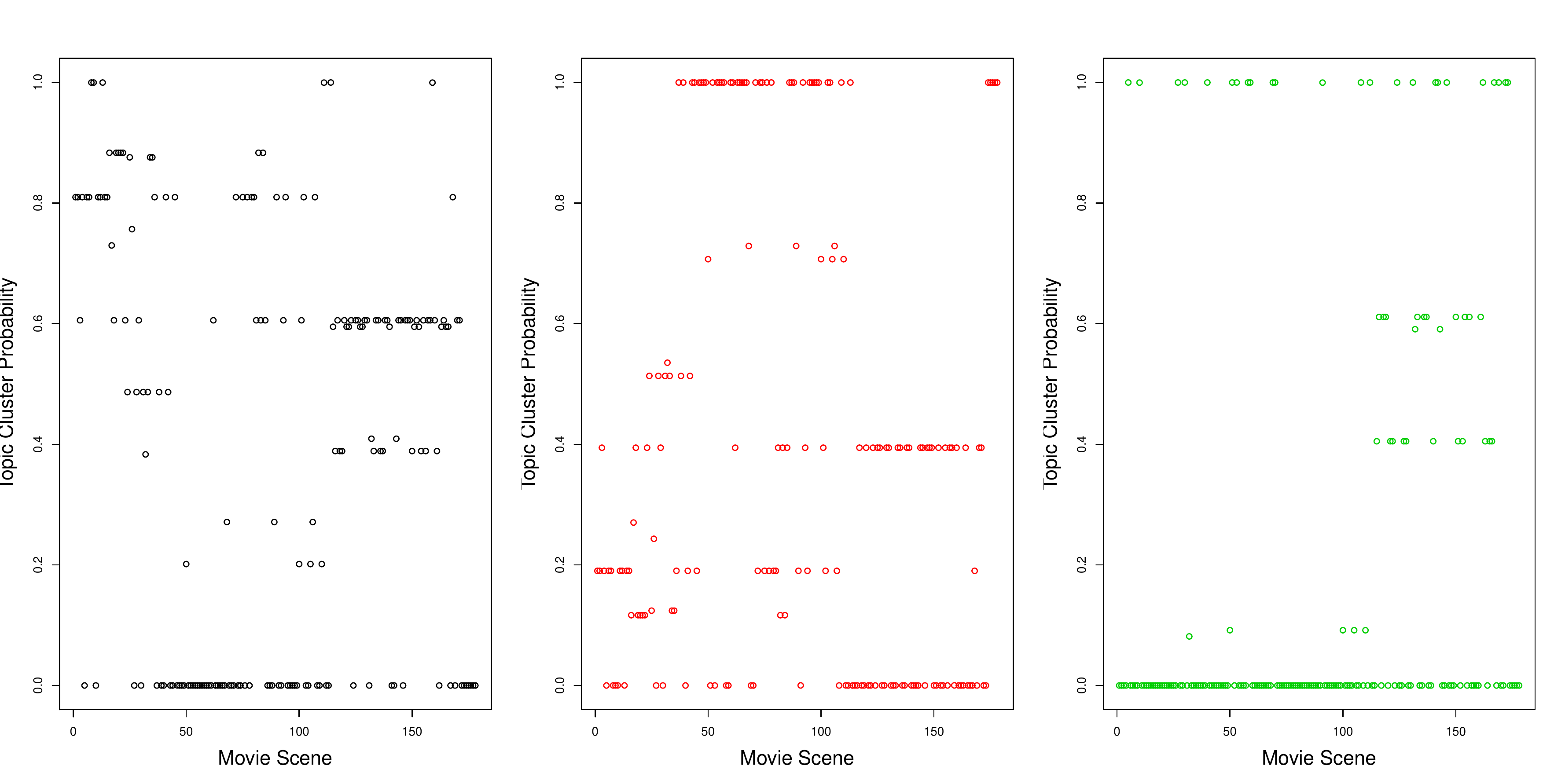}
  \caption{Probability of clusters for movie scenes in Star Wars data set}
    \label{fig:topic_clusters_starwar}
\end{figure}

The estimated cluster assignment probabilities from the EM algorithm for each movie scene in the Star Wars movie are shown in chronological order in Figure~\ref{fig:topic_clusters_starwar}. We can see from the plot that scenes in the early part of the movie are mainly associated with cluster 1, while cluster 2 contains most of the scenes from roughly scene 40 to scene 100. We can deduce from this, for example, that the character ``Han'' is very active in the middle part of the movie. On the other hand, there does not appear to be any obvious pattern for the third cluster. The clustering for many early and late movie scenes is relatively uncertain, as shown in the plot.

\begin{figure}
  \centering
  \includegraphics[width=0.5\textwidth]{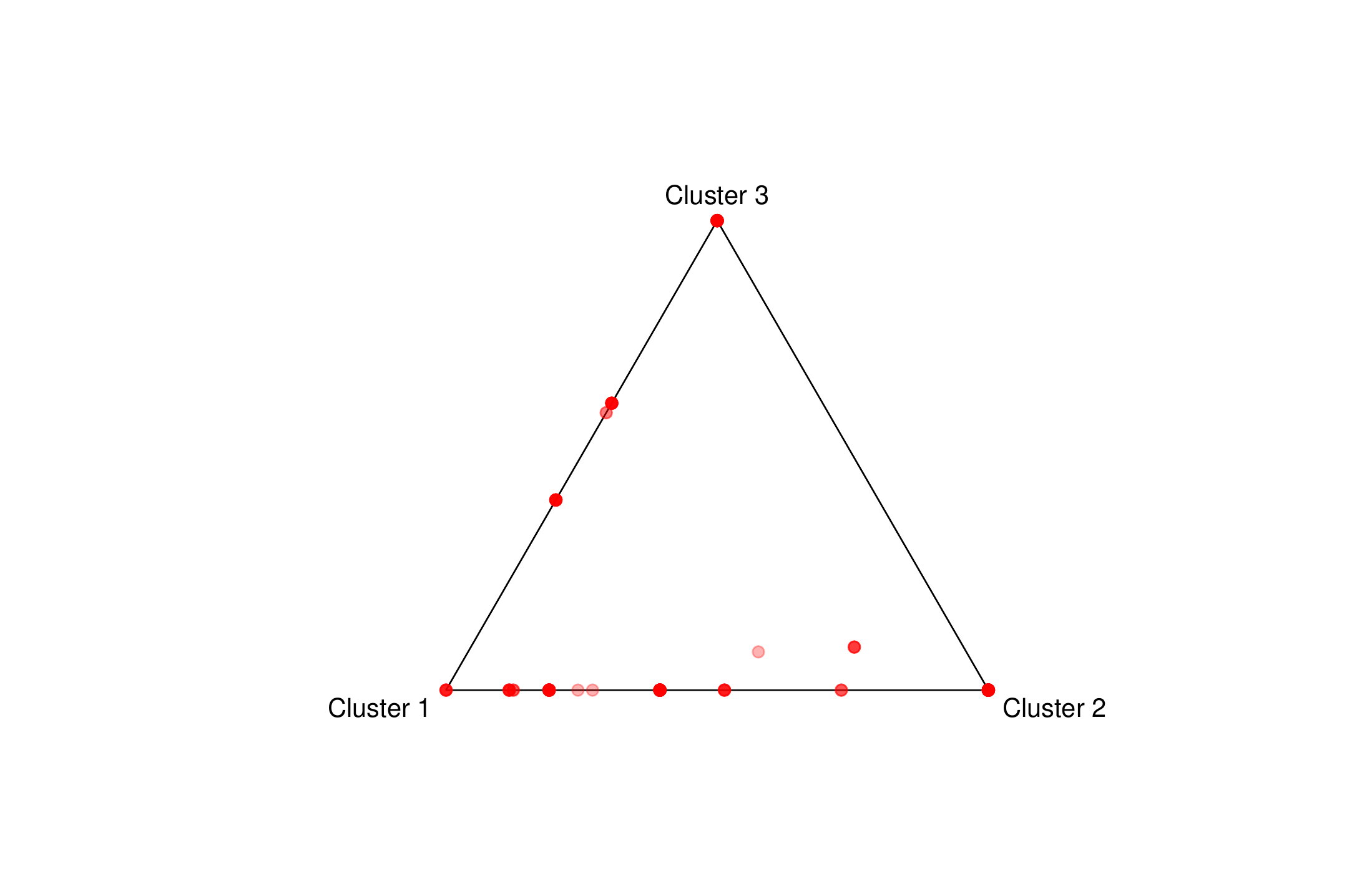}
  \caption{Ternary plot of the {\em a posteriori} group membership probabilities for the scenes in the Star Wars data set}
    \label{fig:topic_clusters_ternary_starwar}
\end{figure}

The uncertainties in clustering are also illustrated in a ternary plot in Figure~\ref{fig:topic_clusters_ternary_starwar}. Each dot in the plot represents a movie scene, and the three corners of the plot represent the three clusters. The closer the dot is to the corner, the higher probability that the corresponding movie scene belongs to the corresponding cluster. The ternary plot in Figure~\ref{fig:topic_clusters_ternary_starwar} shows significant uncertainties in clustering a number of movie scenes into the first two clusters. This is reasonable since for a number of actors including the lead actor ``Luke'', the probabilities of scene appearance are similar for the first two clusters.

\subsection{Lady Gaga Concerts 2014}
As a second application of the ELCA model, we collected the list of songs that Lady Gaga performed in all concerts in 2014 \footnote{\label{gaga} The Lady Gaga setlist data are available at: {\tt http://www.setlist.fm/}  }. The data set contains 96 concerts with a total of 51 distinct songs performed. The hypergraph is constructed by defining each concert as a hyperedge and each song as a vertex. A vertex is contained in a hyperedge if the corresponding song is performed in the corresponding concert. The results of performing model selection using the approach of Section~\ref{sec_cv} is presented in Table~\ref{ms_gaga}. We can see from Table~\ref{ms_gaga} that ELCA models with more than one additional cluster significantly out-perform standard latent class analysis models.

The model with 5 clusters and 2 additional clusters was chosen and fitted to the data set. The parameter estimates $\hat{\pi}$, $\hat{\tau}$ and $\hat{a}$ are given in Table~\ref{para_est_gaga}. We can deduce from $\hat{a}$ and $\hat{\tau}$ that there are a small number of very short concerts of length approximately 14\% of the rest of the ``full'' concerts.

\begin{table}[hbtp]
 \caption{ Estimates of $ \pi $, $ \tau $ and $a$ from fitting the ELCA model with 5 clusters and 2 additional clusters for Lady Gaga concerts 2014 data set}
\begin{center}
\begin{tabular}{| c | c |}
 \hline
$ \hat{\pi} $ & $(0.23, 0.31, 0.23, 0.12, 0.05)$ \\ \hline
  $ \hat{\tau} $ & $(0.11, 0.89)$ \\ \hline
  $ \hat{a} $ & $(0.14, 1)$ \\ \hline
\end{tabular}
\end{center}
  \label{para_est_gaga}
\end{table}

Table~\ref{gaga_phi} shows the parameter estimates $\hat{\phi}$ where the popularity of the 51 songs across 5 clusters are shown. One can see a small number of extremely popular songs which tend to be performed in most concerts, such as ``Paparazzi'', ``Bad Romance'', ``Born This Way'', ``G.U.Y'' and ``Just Dance''. Among the least performed songs, ``Fashion!'' and ``Cake Like Lady Gaga'' tend to be performed in the same concert, while ``Lush Life'', ``It Don't Mean a Thing (If It Ain't Got That Swing)'' and ``But Beautiful'' are more likely to be performed in the same concert. 

\begin{table}[hbtp]
\caption{Estimates of $\{ \phi_{ig} \}$ from fitting the ELCA model with 5 clusters and 2 additional clusters to the Lady Gaga concerts 2014 data set}
 {\scriptsize \resizebox{\textwidth}{!}{
  \begin{tabular}{|c|c|c|c|c|c|}
   \hline
Songs & Cluster 1 & Cluster 2 & Cluster 3 & Cluster 4 & Cluster 5  \\ \hline
Monster for Life	&	0.04	&	0.00	&	0.00	&	0.00	&	0.00	\\ \hline
Fashion!	&	0.00	&	0.00	&	0.68	&	0.00	&	0.00	\\ \hline
Paparazzi	&	1.00	&	1.00	&	1.00	&	1.00	&	1.00	\\ \hline
Bad Romance	&	1.00	&	1.00	&	1.00	&	1.00	&	1.00	\\ \hline
What's Up	&	0.12	&	0.00	&	0.00	&	0.00	&	0.44	\\ \hline
Sophisticated Lady	&	0.00	&	0.00	&	0.00	&	0.00	&	0.31	\\ \hline
Dance in the Dark	&	0.00	&	0.00	&	0.00	&	0.00	&	0.44	\\ \hline
Born This Way	&	1.00	&	1.00	&	1.00	&	1.00	&	1.00	\\ \hline
Judas	&	1.00	&	0.87	&	0.00	&	0.00	&	0.00	\\ \hline
Partynauseous	&	1.00	&	1.00	&	1.00	&	0.00	&	1.00	\\ \hline
Yo and I	&	1.00	&	0.13	&	0.00	&	1.00	&	1.00	\\ \hline
I Will Always Love You	&	0.00	&	0.03	&	0.00	&	0.00	&	0.00	\\ \hline
Monster	&	0.00	&	0.00	&	0.00	&	1.00	&	0.00	\\ \hline
Bang Bang (My Baby Shot Me Down)	&	0.72	&	0.00	&	0.00	&	0.00	&	1.00	\\ \hline
The Queen	&	0.00	&	0.00	&	0.09	&	0.00	&	0.00	\\ \hline
Dope	&	1.00	&	0.60	&	0.00	&	1.00	&	1.00	\\ \hline
Jewels N' Drugs	&	1.00	&	1.00	&	1.00	&	0.00	&	1.00	\\ \hline
Hair	&	0.00	&	0.00	&	0.05	&	0.00	&	0.00	\\ \hline
Mary Jane Holland	&	1.00	&	0.00	&	0.95	&	0.00	&	1.00	\\ \hline
G.U.Y.	&	1.00	&	1.00	&	1.00	&	1.00	&	1.00	\\ \hline
MANiCURE	&	1.00	&	1.00	&	1.00	&	0.00	&	1.00	\\ \hline
Lush Life	&	0.00	&	0.00	&	0.00	&	0.10	&	0.51	\\ \hline
It Don't Mean a Thing (If It Ain't Got That Swing)	&	0.00	&	0.00	&	0.00	&	0.00	&	0.49	\\ \hline
You've Got a Friend	&	0.00	&	0.00	&	0.00	&	0.12	&	0.00	\\ \hline
The Edge of Glory	&	1.00	&	1.00	&	0.04	&	0.00	&	1.00	\\ \hline
Donatella	&	1.00	&	1.00	&	1.00	&	0.00	&	1.00	\\ \hline
But Beautiful	&	0.00	&	0.00	&	0.00	&	0.00	&	0.31	\\ \hline
Do What U Want	&	1.00	&	1.00	&	1.00	&	0.00	&	1.00	\\ \hline
Gypsy	&	1.00	&	1.00	&	1.00	&	0.00	&	1.00	\\ \hline
Applause	&	1.00	&	1.00	&	1.00	&	1.00	&	1.00	\\ \hline
Marry the Night	&	0.04	&	0.00	&	0.05	&	0.00	&	0.44	\\ \hline
Sexxx Dreams	&	1.00	&	0.97	&	1.00	&	1.00	&	1.00	\\ \hline
Another One Bites the Dust	&	0.00	&	0.00	&	0.00	&	0.12	&	0.00	\\ \hline
Just Dance	&	1.00	&	1.00	&	1.00	&	1.00	&	1.00	\\ \hline
Cake Like Lady Gaga	&	0.00	&	0.00	&	0.68	&	0.00	&	0.00	\\ \hline
I Can't Give You Anything but Love, Baby	&	0.00	&	0.03	&	0.00	&	0.00	&	0.49	\\ \hline
Black Jesus  Amen Fashion	&	0.00	&	0.00	&	0.00	&	1.00	&	0.00	\\ \hline
Ratchet	&	1.00	&	1.00	&	1.00	&	0.00	&	1.00	\\ \hline
Aura	&	1.00	&	0.87	&	0.91	&	0.00	&	0.00	\\ \hline
Poker Face	&	1.00	&	1.00	&	1.00	&	1.00	&	1.00	\\ \hline
Venus	&	1.00	&	1.00	&	1.00	&	0.00	&	1.00	\\ \hline
If I Ever Lose My Faith in You	&	0.00	&	0.00	&	0.00	&	0.12	&	0.00	\\ \hline
Bell Bottom Blues	&	0.04	&	0.00	&	0.00	&	0.00	&	0.00	\\ \hline
Whole Lotta Love	&	0.04	&	0.00	&	0.00	&	0.00	&	0.00	\\ \hline
Telephone	&	1.00	&	1.00	&	1.00	&	0.00	&	1.00	\\ \hline
Willkommen	&	0.16	&	0.00	&	0.00	&	0.00	&	0.00	\\ \hline
Brooklyn Nights	&	0.00	&	0.03	&	0.00	&	0.00	&	0.00	\\ \hline
Alejandro	&	1.00	&	1.00	&	1.00	&	0.00	&	1.00	\\ \hline
ARTPOP	&	1.00	&	1.00	&	1.00	&	1.00	&	1.00	\\ \hline
I've Got a Crush on You	&	0.00	&	0.00	&	0.05	&	0.00	&	0.00	\\ \hline
Swine	&	1.00	&	0.97	&	1.00	&	0.00	&	1.00	\\ \hline
  \end{tabular}
}}
\label{gaga_phi}
\end{table}

The estimated cluster assignment probabilities for each concert performed by Lady Gaga in 2014 are shown in chronological order in Figure~\ref{fig:topic_clusters_gaga2014}. There is a strong association between clusters and time of the year. For example, the first 30 concerts performed in 2014 are mainly associated with cluster 1 where songs such as ``Bad Romance'', ``Judas'' and ``Aura'' are among the most popular ones. On the other hand, the next 30 concerts are strongly associated with cluster 2 where songs such as ``The Edge of Glory'', ``Venus'' and ``Ratchet'' are popular. The last 10 concerts of 2014 are mostly clustered into cluster 4 where songs such as ``Yo and I'', ``Monster'' and ``Black Jesus Amen Fashion'' are frequently performed.

\begin{figure}
  \centering
  \includegraphics[width=1.0\textwidth]{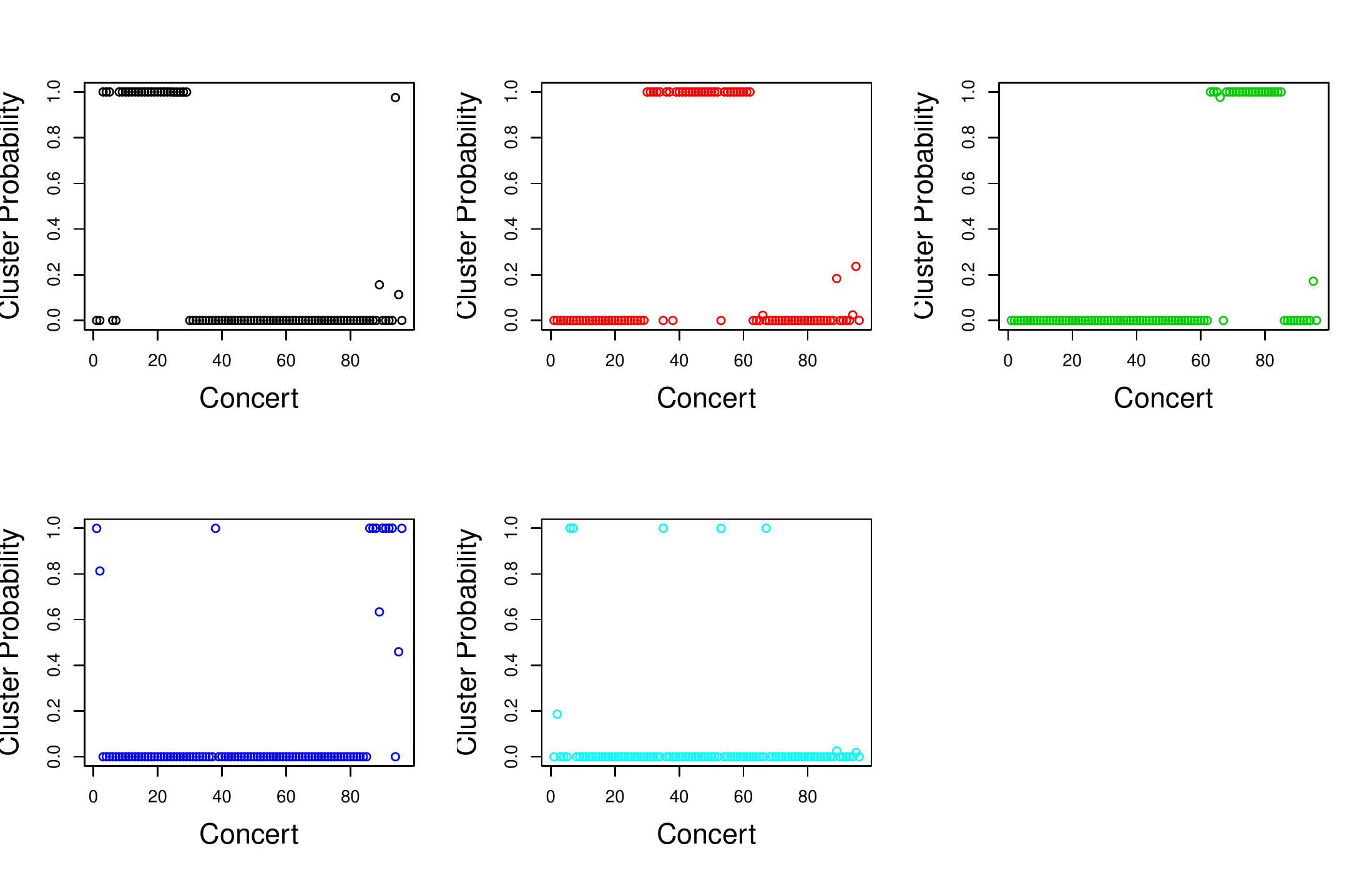}
  \caption{Probability of clusters for concerts in Lady Gaga Concerts 2014 data set}
    \label{fig:topic_clusters_gaga2014}
\end{figure}

Figure~\ref{fig:gaga_size} shows the distribution of hyperedge sizes (or the number of songs performed in concerts) along with the estimated hyperedge sizes by the ELCA model with $G=5$ and $K=2$, and the LCA model with 5 clusters. Adding an additional cluster to the model significantly improves the fit, especially on the tails of the distribution. 

\begin{figure}
  \centering
  \includegraphics[width=0.5\textwidth]{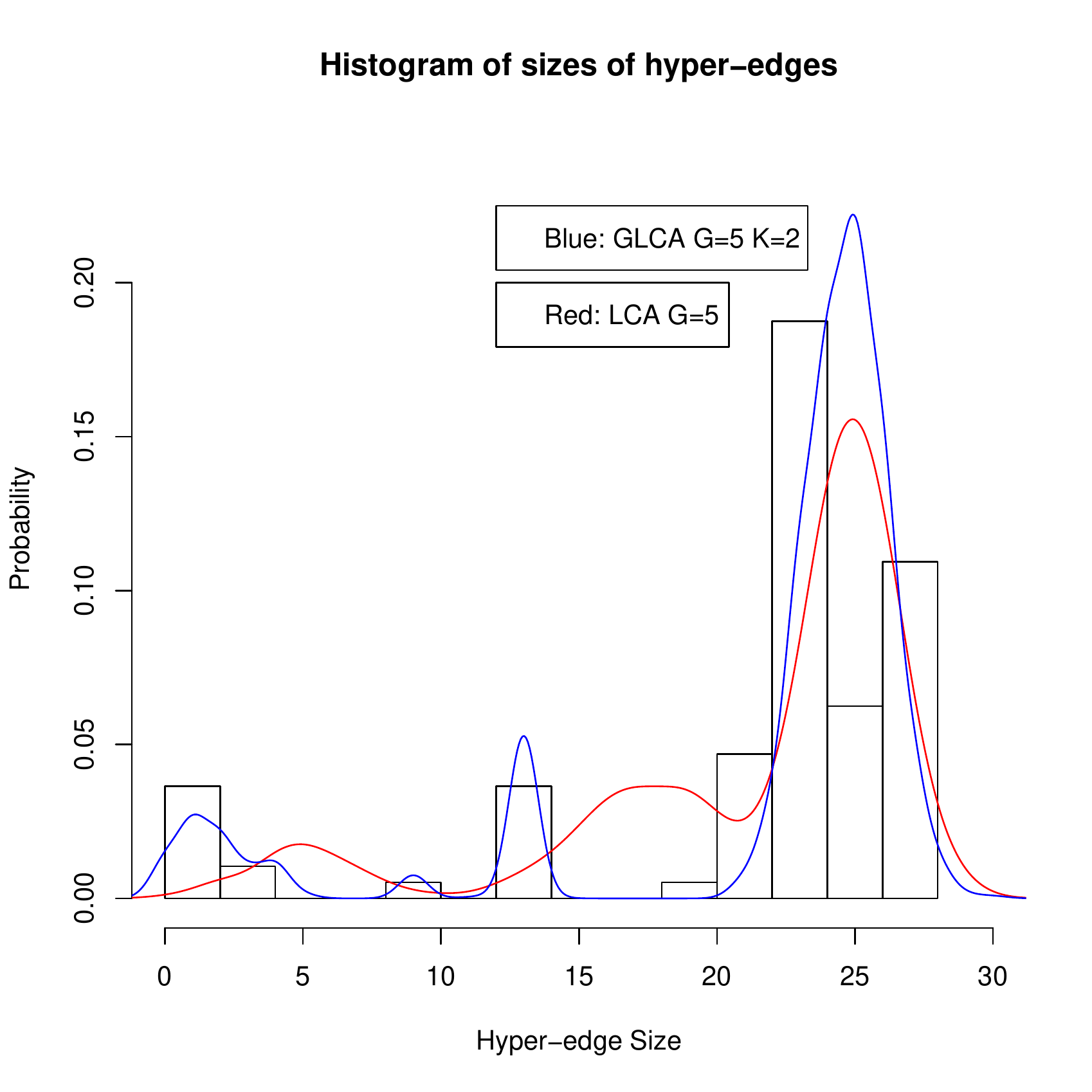}
  \caption{Size distribution of hyperedges of the Lady Gaga Concerts 2014 data set}
    \label{fig:gaga_size}
\end{figure}

\section{Conclusion}
In this paper, we have proposed the Extended Latent Class Analysis model as a generative model for random hypergraphs. The model introduces two clustering structures for hyperedges which captures variation in sizes of hyperedges. 

An EM algorithm has developed for model fitting where the M-step is implemented using a MM algorithm. Model selection is performed using cross validated likelihood method to account for the small sample sizes relative to the number of vertices. 

The model has been shown to give an improved fit relative to the Latent Class Analysis model for three illustrative examples. Furthermore, the fitted model reveals interesting and interpretable structure within the vertices and hyperedges. 

\bibliographystyle{asa}
\bibliography{refs_hgm} 

\begin{thebibliography}{47}
\newcommand{\enquote}[1]{``#1''}
\expandafter\ifx\csname natexlab\endcsname\relax\def\natexlab#1{#1}\fi

\bibitem[{Azondekon et~al.(2018)Azondekon, Harper, Agossa, Welzig, and
  McRoy}]{azondekon18}
Azondekon, R., Harper, Z.~J., Agossa, F.~R., Welzig, C.~M., and McRoy, S.
  (2018), \enquote{Scientific authorship and collaboration network analysis on
  malaria research in Benin: Papers indexed in the Web of Science
  (1996--2016),} \textit{Global Health Research and Policy}, 3, 11.

\bibitem[{Borgatti and Everett(1997)}]{borgatti97}
Borgatti, S.~P. and Everett, M.~G. (1997), \enquote{Network analysis of 2-mode
  data,} \textit{Soc. Networks}, 19, 243 -- 269.

\bibitem[{Celeux and Govaert(1991)}]{celeux91}
Celeux, G. and Govaert, G. (1991), \enquote{Clustering criteria for discrete
  data and latent class models,} \textit{Journal of Classification}, 8,
  157--176.

\bibitem[{Cheng and Church(2000)}]{cheng00}
Cheng, Y. and Church, G.~M. (2000), \enquote{Biclustering of expression data,}
  in \textit{Proceedings of the Eighth International Conference on Intelligent
  Systems for Molecular Biology}, AAAI Press, pp. 93--103.

\bibitem[{de~Panafieu(2015)}]{panafieu15}
de~Panafieu, {\'E}. (2015), \enquote{Phase transition of random non-uniform
  hypergraphs,} \textit{J. Discrete Algorithms}, 31, 26--39.

\bibitem[{Dempster et~al.(1977)Dempster, Laird, and Rubin}]{dempster77}
Dempster, A.~P., Laird, N.~M., and Rubin, D.~B. (1977), \enquote{Maximum
  likelihood from incomplete data via the {EM} algorithm,} \textit{J. Roy.
  Statist. Soc. Ser. B}, 39, 1--38, with discussion.

\bibitem[{Dhillon et~al.(2003)Dhillon, Mallela, and Modha}]{dhillon03}
Dhillon, I.~S., Mallela, S., and Modha, D.~S. (2003),
  \enquote{Information-theoretic co-clustering,} in \textit{Proceedings of the
  Ninth ACM SIGKDD International Conference on Knowledge Discovery and Data
  Mining}, pp. 89--98.

\bibitem[{Doreian and Batagelj(2004)}]{doreian04}
Doreian, P. and Batagelj, V. (2004), \enquote{Generalized blockmodeling of
  two-mode network data,} \textit{Soc. Networks}, 29--53.

\bibitem[{Dyer et~al.(2015)Dyer, Frieze, and Greenhill}]{dyer15}
Dyer, M., Frieze, A., and Greenhill, C. (2015), \enquote{On the chromatic
  number of a random hypergraph,} \textit{J. Combin. Theory Ser. B}, 113,
  68--122.

\bibitem[{Faust et~al.(2002)Faust, Willert, Rowlee, and Skvoretz}]{faust02}
Faust, K., Willert, K., Rowlee, D., and Skvoretz, J. (2002), \enquote{Scaling
  and statistical models for affiliation networks: patterns of participation
  among Soviet politicians during the Brezhnev era,} \textit{Soc. Networks},
  24, 231--259.

\bibitem[{Field et~al.(2006)Field, Frank, Schiller, Riegle-Crumb, and
  Muller}]{field06}
Field, S., Frank, K.~A., Schiller, K., Riegle-Crumb, C., and Muller, C. (2006),
  \enquote{Identifying positions from affiliation networks: Preserving the
  duality of people and events,} \textit{Soc. Networks}, 28, 97 -- 123.

\bibitem[{Friel et~al.(2016)Friel, Rastelli, Wyse, and Raftery}]{friel16}
Friel, N., Rastelli, R., Wyse, J., and Raftery, A.~E. (2016),
  \enquote{Interlocking directorates in Irish companies using a latent space
  model for bipartite networks,} \textit{Proc. Natl. Acad. Sci. U.S.A.}, 113,
  6629--6634.

\bibitem[{Fujimoto et~al.(2011)Fujimoto, Chou, and Valente}]{fujimoto11}
Fujimoto, K., Chou, C.-P., and Valente, T.~W. (2011), \enquote{The network
  autocorrelation model using two-mode data: Affiliation exposure and potential
  bias in the autocorrelation parameter,} \textit{Soc. Networks}, 33, 231 --
  243.

\bibitem[{George and Merugu(2005)}]{george05}
George, T. and Merugu, S. (2005), \enquote{A scalable collaborative filtering
  framework based on co-clustering,} in \textit{Proceedings of the Fifth IEEE
  International Conference on Data Mining}, pp. 625--628.

\bibitem[{Goldschmidt(2005)}]{goldschmidt05}
Goldschmidt, C. (2005), \enquote{Critical random hypergraphs: the emergence of
  a giant set of identifiable vertices,} \textit{Ann. Probab.}, 33, 1573--1600.

\bibitem[{Goodman(1974)}]{goodman74}
Goodman, L.~A. (1974), \enquote{Exploratory latent structure analysis using
  both identifiable and unidentifiable models,} \textit{Biometrika}, 61,
  215--231.

\bibitem[{Govaert and Nadif(2003)}]{govaert03}
Govaert, G. and Nadif, M. (2003), \enquote{Clustering with block mixture
  models.} \textit{Pattern Recognition}, 36, 463--473.

\bibitem[{Govaert and Nadif(2008)}]{govaert08}
--- (2008), \enquote{Block clustering with {B}ernoulli mixture models:
  comparison of different approaches,} \textit{Comput. Statist. Data Anal.},
  52, 3233--3245.

\bibitem[{Handcock et~al.(2007)Handcock, Raftery, and Tantrum}]{handcock07}
Handcock, M.~S., Raftery, A.~E., and Tantrum, J.~M. (2007),
  \enquote{Model-based clustering for social networks,} \textit{J. Roy.
  Statist. Soc. Ser. A}, 170, 301--354.

\bibitem[{Hoff et~al.(2002)Hoff, Raftery, and Handcock}]{hoff02}
Hoff, P.~D., Raftery, A.~E., and Handcock, M.~S. (2002), \enquote{Latent space
  approaches to social network analysis,} \textit{J. Amer. Statist. Assoc.},
  97, 1090--1098.

\bibitem[{Holland and Leinhardt(1981)}]{holland81}
Holland, P.~W. and Leinhardt, S. (1981), \enquote{An exponential family of
  probability distributions for directed graphs,} \textit{J. Amer. Statist.
  Assoc.}, 76, 33--65.

\bibitem[{Hunter and Lange(2004)}]{hunter04}
Hunter, D.~R. and Lange, K. (2004), \enquote{A tutorial on {MM} algorithms,}
  \textit{Amer. Statist.}, 58, 30--37.

\bibitem[{Karo{\'n}ski and {\L}uczak(2002)}]{karonski02}
Karo{\'n}ski, M. and {\L}uczak, T. (2002), \enquote{The phase transition in a
  random hypergraph,} \textit{J. Comput. Appl. Math.}, 142, 125--135.

\bibitem[{Koskinen and Edling(2012)}]{koskinen12}
Koskinen, J. and Edling, C. (2012), \enquote{Modelling the evolution of a
  bipartite network - Peer referral in interlocking directorates,} \textit{Soc.
  Networks}, 34, 309 -- 322, dynamics of Social Networks (2).

\bibitem[{Lange et~al.(2000)Lange, Hunter, and Yang}]{lange00}
Lange, K., Hunter, D.~R., and Yang, I. (2000), \enquote{Optimization transfer
  using surrogate objective functions,} \textit{J. Comput. Graph. Statist.}, 9,
  1--59.

\bibitem[{Latapy et~al.(2008)Latapy, Magnien, and Vecchio}]{latapy08}
Latapy, M., Magnien, C., and Vecchio, N.~D. (2008), \enquote{Basic notions for
  the analysis of large two-mode networks,} \textit{Soc. Networks}, 30, 31--48.

\bibitem[{Latouche et~al.(2011)Latouche, Birmel\'e, and Ambroise}]{latouche11}
Latouche, P., Birmel\'e, E., and Ambroise, C. (2011), \enquote{Overlapping
  stochastic block models with application to the {F}rench political
  blogosphere,} \textit{Ann. Appl. Stat.}, 5, 309--336.

\bibitem[{Lazarsfeld and Henry(1968)}]{lazarsfeld68}
Lazarsfeld, P.~F. and Henry, N.~W. (1968), \textit{Latent structure analysis},
  Houghton Mifflin, Boston, MA 02110, USA.

\bibitem[{Lind et~al.(2005)Lind, Gonz\'alez, and Herrmann}]{lind05}
Lind, P.~G., Gonz\'alez, M.~C., and Herrmann, H.~J. (2005), \enquote{Cycles and
  clustering in bipartite networks,} \textit{Phys. Rev. E}, 72, 056127.

\bibitem[{Lunagómez et~al.(2017)Lunagómez, Mukherjee, Wolpert, and
  Airoldi}]{lunagomez17}
Lunagómez, S., Mukherjee, S., Wolpert, R.~L., and Airoldi, E.~M. (2017),
  \enquote{Geometric representations of random hypergraphs,} \textit{J. Amer.
  Stat. Assoc.}, 112, 363--383.

\bibitem[{Meng and Rubin(1993)}]{meng93}
Meng, X.-L. and Rubin, D.~B. (1993), \enquote{Maximum likelihood estimation via
  the {ECM} algorithm: a general framework,} \textit{Biometrika}, 80, 267--278.

\bibitem[{Moody(2004)}]{moody04}
Moody, J. (2004), \enquote{The structure of a social science collaboration
  network: disciplinary cohesion from 1963 to 1999,} \textit{Am. Sociol. Rev},
  69, 213--238.

\bibitem[{Newman(2004)}]{newman04}
Newman, M.~E. (2004), \textit{Who Is the Best Connected Scientist? A Study of
  Scientific Coauthorship Networks}, Berlin, Heidelberg: Springer Berlin
  Heidelberg, pp. 337--370.

\bibitem[{Newman(2001{\natexlab{a}})}]{newman01a}
Newman, M. E.~J. (2001{\natexlab{a}}), \enquote{Scientific collaboration
  networks. I. Network construction and fundamental results,} \textit{Phys.
  Rev. E}, 64, 016131.

\bibitem[{Newman(2001{\natexlab{b}})}]{newman01b}
--- (2001{\natexlab{b}}), \enquote{Scientific collaboration networks. II.
  Shortest paths, weighted networks, and centrality,} \textit{Phys. Rev. E},
  64, 016132.

\bibitem[{Nowicki and Snijders(2001)}]{nowicki01}
Nowicki, K. and Snijders, T. A.~B. (2001), \enquote{Estimation and prediction
  for stochastic blockstructures,} \textit{J. Amer. Statist. Assoc.}, 96,
  1077--1087.

\bibitem[{Perugini et~al.(2004)Perugini, Gon{\c{c}}alves, and Fox}]{perugini04}
Perugini, S., Gon{\c{c}}alves, M.~A., and Fox, E.~A. (2004),
  \enquote{Recommender systems research: a connection-centric survey,}
  \textit{Journal of Intelligent Information Systems}, 23, 107--143.

\bibitem[{Poole(2015)}]{poole15}
Poole, D. (2015), \enquote{On the strength of connectedness of a random
  hypergraph,} \textit{Electron. J. Combin.}, 22, Paper 1.69, 16.

\bibitem[{Rau et~al.(2015)Rau, Maugis-Rabusseau, Martin-Magniette, and
  Celeux}]{rau15}
Rau, A., Maugis-Rabusseau, C., Martin-Magniette, M.-L., and Celeux, G. (2015),
  \enquote{Co-expression analysis of high-throughput transcriptome sequencing
  data with Poisson mixture models,} \textit{Bioinformatics}, 31, 1420.

\bibitem[{Robins and Alexander(2004)}]{robins04}
Robins, G. and Alexander, M. (2004), \enquote{Small worlds among interlocking
  directors: network structure and distance in bipartite graphs,}
  \textit{‎Comput. Math. Organ. Theory}, 10, 69--94.

\bibitem[{Skvoretz and Faust(1999)}]{skvoretz99}
Skvoretz, J. and Faust, K. (1999), \enquote{Logit models for affiliation
  networks,} \textit{Sociol. Methodol}, 29, 253--280.

\bibitem[{Smyth(2000)}]{smyth00}
Smyth, P. (2000), \enquote{Model selection for probabilistic clustering using
  cross-validated likelihood,} \textit{Statistics and Computing}, 10, 63--72.

\bibitem[{Snijders et~al.(2013)Snijders, Lomi, and Torló}]{snijders13}
Snijders, T.~A., Lomi, A., and Torló, V.~J. (2013), \enquote{A model for the
  multiplex dynamics of two-mode and one-mode networks, with an application to
  employment preference, friendship, and advice,} \textit{Soc. Networks}, 35,
  265--276.

\bibitem[{Stasi et~al.(2014)Stasi, Sadeghi, Rinaldo, Petrovic, and
  Fienberg}]{stasi14}
Stasi, D., Sadeghi, K., Rinaldo, A., Petrovic, S., and Fienberg, S. (2014),
  \enquote{{$\beta$} models for random hypergraphs with a given degree
  sequence,} in \textit{Proceedings of {COMPSTAT} 2014---21st {I}nternational
  {C}onference on {C}omputational {S}tatistics}, pp. 593--600.

\bibitem[{Wang et~al.(2013)Wang, Pattison, and Robins}]{wang13}
Wang, P., Pattison, P., and Robins, G. (2013), \enquote{Exponential random
  graph model specifications for bipartite networks - A dependence hierarchy,}
  \textit{Soc. Networks}, 35, 211--222.

\bibitem[{Wang et~al.(2009)Wang, Sharpe, Robins, and Pattison}]{wang09}
Wang, P., Sharpe, K., Robins, G., and Pattison, P. (2009), \enquote{Exponential
  random graph (p*) models for affiliation networks,} \textit{Soc. Networks},
  31, 12--25.

\bibitem[{Wang(1993)}]{wang93}
Wang, Y.~H. (1993), \enquote{On the number of successes in independent trials,}
  \textit{Statist. Sinica}, 3, 295--312.

\end{thebibliography}

\appendix

\section{Proof on Proposition~\ref{comp_glca}}
\begin{proof}
   We can write $A = \sum_{i=1}^{N} A_{i}$ where $A_{i}=1$ if node $i$ appears in the hyperedge and $A_{i}=0$ otherwise. Similarly, we write $B=\sum_{i=1}^{N} B_{i}$. Let $Z_{A}$ be the latent cluster assignment of $X_{A}$ where $Z_{A} = g$ if $X_{A}$ is generated from cluster $g$. Let $Z_{B}^{(1)}$ and $Z_{B}^{(2)}$ be the latent cluster and additional clusters assignments of $X_{B}$, where $Z_{B}^{(1)}=g$ and $Z_{B}^{(2)}=k$ if $X_{B}$ is generated from cluster $g$ and additional clusters $k$.
We have
\begin{eqnarray*}
  E(A) &=& \sum_{g=1}^{G} E(A|Z_{A}) Pr(Z_{A}=g) \\
         &=& \sum_{g=1}^{G} \sum_{i=1}^{N} E(A_{i} |Z_{A}=g) Pr(Z_{A}=g) \\
         &=& \sum_{g=1}^{G} \sum_{i=1}^{N} p_{ig} \pi_{g}
\end{eqnarray*}
\begin{eqnarray*}
   E(B) &=& \sum_{g=1}^{G} \sum_{k=1}^{K} E(B|Z_{B}^{(1)}=g, Z_{B}^{(2)}=k) Pr(Z_{B}^{(1)}=g, Z_{B}^{(2)}=k) \\
         &=& \sum_{g=1}^{G} \sum_{k=1}^{K} \sum_{i=1}^{N} E(B_{i}|Z_{B}^{(1)}=g, Z_{B}^{(2)}=k) Pr(Z_{B}^{(1)}=g, Z_{B}^{(2)}=k) \\
         &=& \sum_{g=1}^{G} \sum_{k=1}^{K} \sum_{i=1}^{N} \phi_{ig} a_{k} \tau_{k} \pi_{g} \\
         &=& \sum_{g=1}^{G} \sum_{i=1}^{N} p_{ig} \pi_{g} \\
         &=& E(A)
\end{eqnarray*}
For the variance of the LCA model, we have that
\begin{eqnarray*}
  Var(A) = \sum_{i=1}^{N} Var( A_{i}) + 2 \sum_{i<j}^{N} Cov(A_{i}, A_{j})
\end{eqnarray*}
where
\begin{eqnarray*}
  Var(A_{i}) &=& E(A_{i}^{2}) - E(A_{i})^{2} \\
                  &=& Pr(A_{i}=1) - Pr(A_{i}=1)^{2} \\
                  &=& \sum_{g=1}^{G} p_{ig} \pi_{g} - \Big( \sum_{g=1}^{G} p_{ig} \pi_{g} \Big)^{2} \\
\end{eqnarray*}
\begin{eqnarray*}
   Cov(A_{i}, A_{j}) &=& E(A_{i} A_{j}) - E(A_{i}) E(A_{j})  \\
                   &=& Pr(A_{i}=A_{j}=1) - Pr(A_{i}=1) Pr(A_{j}=1) \\
                  &=& \sum_{g=1}^{G} p_{ig} p_{jg} \pi_{g} - \Big( \sum_{g=1}^{G} p_{ig} \pi_{g} \Big) \Big( \sum_{g=1}^{G} p_{jg} \pi_{g} \Big)
\end{eqnarray*}
Hence, we have that
\begin{eqnarray*}
   Var(A) = \sum_{i=1}^{N} \sum_{g=1}^{G} p_{ig} \pi_{g} - \sum_{i=1}^{N} \Big( \sum_{g=1}^{G} p_{ig} \pi_{g} \Big)^{2} \\
                   + 2 \sum_{i<j}^{N} \sum_{g=1}^{G} p_{ig} p_{jg} \pi_{g} - 2 \sum_{i<j}^{N} \Big( \sum_{g=1}^{G} p_{ig} \pi_{g} \Big) \Big( \sum_{g=1}^{G} p_{jg} \pi_{g} \Big) 
\end{eqnarray*}
Now, 
\begin{eqnarray*}
  Var(B) = \sum_{i=1}^{N} Var(B_{i}) + 2 \sum_{i<j}^{N} Cov(B_{i}, B_{j}) 
\end{eqnarray*}
\begin{eqnarray*}
  Var(B_{i}) &=& Pr(B_{i}=1) - Pr(B_{i}=1)^{2} \\
        &=& \sum_{g=1}^{G} \sum_{k=1}^{K} \phi_{ig} a_{k} \tau_{k} \pi_{g} - \Big( \sum_{g=1}^{G}  \sum_{k=1}^{K} \phi_{ig} a_{k} \tau_{k} \pi_{g} \Big)^{2} \\
       &=& \sum_{g=1}^{G} p_{ig} \pi_{g} - \Big( \sum_{g=1}^{G} p_{ig} \pi_{g} \Big)^{2} 
\end{eqnarray*}
\begin{eqnarray*}
   Cov(B_{i}, B_{j}) &=& Pr(B_{i}=B_{j}=1) - Pr(B_{i}=1) Pr(B_{j}=1) \\
                            &=& \sum_{g=1}^{G} \sum_{k=1}^{K} \phi_{ig} \phi_{jg} a_{k}^{2} \pi_{g} \tau_{k} - \Big( \sum_{g=1}^{G} p_{ig} \pi_{g} \Big) \Big(\sum_{g=1}^{G} p_{jg} \pi_{g} \Big) \\
\end{eqnarray*}
We have
\begin{eqnarray*}
   Var(B) = \sum_{i=1}^{N} \sum_{g=1}^{G} p_{ig} \pi_{g} - \sum_{i=1}^{N} \Big( \sum_{g=1}^{G} p_{ig} \pi_{g} \Big)^{2} \\
                 + 2 \sum_{i<j}^{N} \sum_{g=1}^{G} \sum_{k=1}^{K} \phi_{ig} \phi_{jg} a_{k}^{2} \pi_{g} \tau_{k} - 2 \sum_{i<j}^{N} \Big( \sum_{g=1}^{G} p_{ig} \pi_{g} \Big) \Big( \sum_{g=1}^{G} p_{jg} \pi_{g} \Big)
\end{eqnarray*}
Now, 
\begin{eqnarray*}
  Var(B) - Var(A) &=& 2 \sum_{i<j}^{N} \sum_{g=1}^{G} \sum_{k=1}^{K} \phi_{ig} \phi_{jg} a_{k}^{2} \pi_{g} \tau_{k} -  2 \sum_{i<j}^{N} \sum_{g=1}^{G} p_{ig} p_{jg} \pi_{g} \\
                    &=& 2 \sum_{i<j}^{N} \sum_{g=1}^{G} \Big( \sum_{k=1}^{K} \phi_{ig} \phi_{jg} a_{k}^{2} \tau_{k} - p_{ig} p_{jg} \Big) \pi_{g} \\
              &=& 2 \sum_{i<j}^{N} \sum_{g=1}^{G} \phi_{ig} \phi_{jg} \Big( \sum_{k=1}^{K}  a_{k}^{2} \tau_{k} - \Big(\sum_{k=1}^{K} a_{k} \tau_{k} \Big)^{2} \Big) \pi_{g}
\end{eqnarray*}
To show the quantity above is non-negative, we have to show that 
\begin{eqnarray*}
   \sum_{k=1}^{K}  a_{k}^{2} \tau_{k} - \Big(\sum_{k=1}^{K} a_{k} \tau_{k} \Big)^{2} \ge 0
\end{eqnarray*}  
which follows from Jensen's inequality.
\end{proof}

\section{M-step of EM Algorithm}
\label{sec_em_detail}
For the M-step, we need to maximize $Q(\theta|\theta^{(t)})$ with respect to the model parameters $\{ \phi_{ig} \}$, $\{ a_{k} \}$, $\{ \pi_{g} \}$ and $\{ \tau_{k} \}$.

\subsection{Maximize w.r.t. $\phi_{ig}$}
For fixed $i$ and $g$, the objective function retaining terms involving $\phi_{ig}$ can be written as
\begin{eqnarray}
 \label{obj_phi}
  Q = \sum_{j=1}^{M} \sum_{k=1}^{K} \widehat{Z^{(1)}_{jg} Z^{(2)}_{jk}} \Big( x_{ij} \log(\phi_{ig}) + (1-x_{ij}) \log(1-a_{k}\phi_{ig}) \Big) 
\end{eqnarray}
Since an analytic expression for $ \argmax_{\phi_{ig}}\{Q\}$ does not exist due to the term $ \log(1-a_{k}\phi_{ig})$, we apply the MM (Minorization Maximization) algorithm \citep{hunter04}. We first apply a quadratic lower bound on the concave function $\log(1-a_{k}\phi_{ig})$  for $k < K$. We let 
\begin{eqnarray*} 
  f(\phi_{ig}) = \log(1-a_{k}\phi_{ig}).
\end{eqnarray*}
We then have 
\begin{eqnarray*}
  \frac{\partial f}{\partial \phi_{ig}} = \frac{-a_{k}}{1-a_{k} \phi_{ig}}
\end{eqnarray*}
\begin{eqnarray*}
  \frac{\partial^{2} f}{\partial \phi_{ig}^{2}} = \frac{-a_{k}^{2}}{(1-a_{k} \phi_{ig})^{2}} \ge \frac{ -a_{k}^{2} }{(1-a_{k})^{2}}
\end{eqnarray*}
Hence, we have
\begin{eqnarray*}
  \log(1-a_{k} \phi_{ig}) \ge \log(1-a_{k} \phi^{(t)}_{ig}) + \Big( \frac{-a_{k}}{1-a_{k}\phi^{(t)}_{ig}} \Big) (\phi_{ig}-\phi^{(t)}_{ig}) + \frac{1}{2} \Big( \frac{-a_{k}^{2}}{(1-a_{k})^{2}} \Big) (\phi_{ig} - \phi^{(t)}_{ig})^{2}
\end{eqnarray*}

Hence, the objective function in (\ref{obj_phi}) up to an additive constant can be minorized by the function below.
\begin{eqnarray}
 \label{lower_phi}
  Q_{lower} = \sum_{j=1}^{M} \sum_{k=1}^{K} \widehat{Z^{(1)}_{jg} Z^{(2)}_{jk}} x_{ij} \log(\phi_{ig})  + \sum_{j=1}^{M} \sum_{k=1}^{K-1} \widehat{Z^{(1)}_{jg} Z^{(2)}_{jk}} (1-x_{ij}) \\ \ \nonumber
  \Bigg( \Big( \frac{-a_{k}}{1-a_{k}\phi^{(t)}_{ig}} \Big) \phi_{ig} + \frac{1}{2} \Big( \frac{-a_{k}^{2}}{(1-a_{k})^{2}} \Big) (\phi_{ig} - \phi^{(t)}_{ig})^{2} \Bigg) \\ \nonumber
  + \sum_{j=1}^{M} \widehat{Z_{jg}^{(1)} Z_{jK}^{(2)} } (1-x_{ij}) \log(1-\phi_{ig})
\end{eqnarray}

To simplify the expression above, we define the quantities below.
\begin{eqnarray*}
   A_{1} = \sum_{j=1}^{M} \sum_{k=1}^{K} \widehat{Z^{(1)}_{jg} Z^{(2)}_{jk}} x_{ij}
\end{eqnarray*} 
\begin{eqnarray*}
  A_{2} = \sum_{j=1}^{M} \widehat{Z^{(1)}_{jg} Z^{(2)}_{jK}} (1-x_{ij})
\end{eqnarray*}
\begin{eqnarray*} 
   B_{1}=  \sum_{j=1}^{M} \sum_{k=1}^{K-1} \widehat{Z^{(1)}_{jg} Z^{(2)}_{jk}} (1-x_{ij}) \frac{-a_{k}}{1-a_{k} \phi_{ig}^{(t)}} 
\end{eqnarray*}
\begin{eqnarray*}
  B_{2} =  \sum_{j=1}^{M} \sum_{k=1}^{K-1} \widehat{Z^{(1)}_{jg} Z^{(2)}_{jk}} (1-x_{ij}) \frac{1}{2} \frac{-a_{k}^{2}}{(1-a_{k})^{2}}
\end{eqnarray*}

Now, the lower bound in (\ref{lower_phi}) can be written as below.
\begin{eqnarray*}
  Q_{lower} = A_{1} \log(\phi_{ig}) + A_{2} \log(1-\phi_{ig}) + B_{1} \phi_{ig} + B_{2} (\phi_{ig}-\phi_{ig}^{(t)})^{2}
\end{eqnarray*}
Taking derivative with respect to $\phi_{ig}$, we have
\begin{eqnarray*}
  \frac{A_{1}}{\phi_{ig}} - \frac{A_{2}}{1-\phi_{ig}} + B_{1} + 2B_{2}\phi_{ig} -2B_{2}\phi_{ig}^{(t)} = 0
\end{eqnarray*}
Let $C=B_{1} - 2B_{2} \phi^{(t)}_{ig}$, we have
\begin{eqnarray}
\label{update_phi}
  \phi_{ig}^{3} - \frac{2B_{2}-C}{2B_{2}} \phi_{ig}^{2} - \frac{C-A_{1}-A_{2}}{2B_{2}} \phi_{ig} - \frac{A_{1}}{2B_{2}} = 0
\end{eqnarray}
Solving the cubic equation above results in the update for $\phi_{ig}$.

\subsection{Maximize w.r.t. $a_{k}$}
For a fixed $k$, the objective function (\ref{glca_ll}) retaining terms involving $a_{k}$ can be expressed as
\begin{eqnarray}
 \label{obj_a}
  Q = \sum_{j=1}^{M} \sum_{g=1}^{G} \widehat{Z^{(1)}_{jg} Z^{(2)}_{jk}} \Big( \sum_{i=1}^{N} x_{ij} \log(a_{k}) + (1-x_{ij}) \log(1-a_{k}\phi_{ig}) \Big) .
\end{eqnarray}
Since an analytic expression for $ \argmax_{a_{k}}\{Q\}$ does not exist due to the term $ \log(1-a_{k}\phi_{ig})$, we apply the MM (Minorization Maximization) algorithm. We first apply a quadratic lower bound on the concave function 
\begin{eqnarray*}
   \log(1-a_{k} \phi_{ig}) \ge \log(1-a^{(t)}_{k} \phi_{ig}) + \Big( \frac{-\phi_{ig}}{1-a_{k}^{(t)} \phi_{ig}} \Big) (a_{k}-a_{k}^{(t)}) + \frac{1}{2} \Big( \frac{-\phi_{ig}^{2}}{(1-\phi_{ig})^{2} }\Big) (a_{k}-a_{k}^{(t)})^{2} 
\end{eqnarray*}

Hence, (\ref{obj_a}) up to an additive constant can be minorized by the function below.
\begin{eqnarray}
\label{a_lower}
  Q_{lower} = \Bigg( \sum_{j=1}^{M} \sum_{g=1}^{G} \widehat{Z_{jg}^{(1)} Z_{jk}^{(2)}} \sum_{i=1}^{N} x_{ij} \Bigg) \log(a_{k}) +  \sum_{j=1}^{M} \sum_{g=1}^{G} \widehat{Z_{jg}^{(1)} Z_{jk}^{(2)}} \sum_{i=1}^{N} (1-x_{ij}) \\ \nonumber
 \Bigg( \frac{-\phi_{ig}}{1-a_{k}^{(t)} \phi_{ig}} a_{k} + \frac{1}{2} \Big(\frac{-\phi_{ig}^{2}}{(1-\phi_{ig})^{2}} \Big)(a_{k}-a_{k}^{(t)})^{2} \Bigg) 
\end{eqnarray}

To simply the expression above, we define the following quantities.
\begin{eqnarray*}
  A = \sum_{j=1}^{M} \sum_{g=1}^{G} \widehat{Z_{jg}^{(1)} Z_{jk}^{(2)}} \sum_{i=1}^{N} x_{ij}
\end{eqnarray*}
\begin{eqnarray*}
  B = \sum_{j=1}^{M} \sum_{g=1}^{G} \widehat{Z_{jg}^{(1)} Z_{jk}^{(2)}} \sum_{i=1}^{N} (1-x_{ij}) \Big( \frac{-\phi_{ig}}{1-a_{k}^{(t)} \phi_{ig}} \Big)
\end{eqnarray*}
\begin{eqnarray*}
  C = \sum_{j=1}^{M} \sum_{g=1}^{G} \widehat{Z_{jg}^{(1)} Z_{jk}^{(2)}} \sum_{i=1}^{N} (1-x_{ij}) \frac{1}{2} \Big(\frac{-\phi_{ig}^{2}}{(1-\phi_{ig})^{2}} \Big)
\end{eqnarray*}
Taking derivative of (\ref{obj_a}) with respect to $a_{k}$, we have
\begin{eqnarray*}
  \frac{ \partial Q_{lower}}{\partial a_{k}} = \frac{A}{a_{k}} + B + 2C(a_{k}-a_{k}^{(t)}) = 0
\end{eqnarray*}

Let $D = (\frac{B}{2C} - a_{k}^{(t)})$, $E=-\frac{A}{2C}$, we have
\begin{eqnarray}
\label{update_a}
  \hat{a}_{k} = \Big( E+\frac{D^{2}}{4} \Big)^{1/2} - \frac{D}{2}
\end{eqnarray}

\subsection{Maximize w.r.t. $\pi_{g}$ and $\tau_{k}$}
The update for $\pi_{g}$ and $\tau_{k}$ are straightforward and are given below.
\begin{eqnarray}
\label{update_pi}
  \hat{\pi}_{g} \propto \sum_{j=1}^{M} \sum_{k=1}^{K} \widehat{ Z_{jg}^{(1)} Z_{jk}^{(2)} }
\end{eqnarray}
\begin{eqnarray}
\label{update_tau}
  \hat{\tau}_{k} \propto \sum_{j=1}^{M} \sum_{g=1}^{G}  \widehat{ Z_{jg}^{(1)} Z_{jk}^{(2)} }
\end{eqnarray}

\section{Model Selection \& Data Description}
\label{sec_supp}

\setcounter{table}{0}
\renewcommand{\thetable}{A\arabic{table}}

\begin{table}[hbtp]
 \caption{Model Selection for the Star Wars Data Set} 
 \begin{tabular}{| l | l | r|}
   \hline
   \textbf{No. of Clusters} & \textbf{No. of Additional Clusters}  & Cross Validated Loglikelihood \\ \hline
     1 & 1 & -194.77 \\ \hline
     1 & 2 & -206.18 \\ \hline
     2 & 1 & -194.39 \\ \hline
     2 & 2 & -193.92 \\ \hline
     2 & 3 & -194.63 \\ \hline
     3 & 1 & -194.12 \\ \hline
     3 & 2 & \textbf{-190.78} \\ \hline
     3 & 3 & -194.96 \\ \hline
     4 & 1 & -194.99 \\ \hline
  \end{tabular}
    \label{ms_starwars}
\end{table}

\begin{table}[hbtp]
  \caption{Model Selection for Lady Gaga Concerts 2014 Data Set}
 \begin{tabular}{| l | l | r|}
   \hline
   \textbf{No. of Clusters} & \textbf{No. of Additional Clusters}  & Cross Validated Loglikelihood \\ \hline
     1 & 1 &-521.82 \\ \hline
     1 & 2 & -822.68 \\ \hline
     2 & 1 & -389.93 \\ \hline
     2 & 2 & -306.86 \\ \hline
     2 & 3 & -270.79 \\ \hline
     2 & 4 & -274.56 \\ \hline
     3 & 1 & -340.95 \\ \hline
     3 & 2 & -274.87 \\ \hline
     3 & 3 & -261.05 \\ \hline
     3 & 4 & -246.11 \\ \hline
     3 & 5 & -250.71 \\ \hline
     4 & 1 & -338.64 \\ \hline
     4 & 2 & -237.21 \\ \hline
     4 & 3 & -248.05 \\ \hline
     5 & 1 & -322.92 \\ \hline
     5 & 2 & \textbf{-215.39} \\ \hline
     5 & 3 & -236.71 \\ \hline
     6 & 1 & -328.36 \\ \hline
  \end{tabular}
  \label{ms_gaga}
\end{table}

\end{document}